\documentclass[conference]{IEEEtran}

\usepackage{pgfplots}
\usepackage[utf8]{inputenc}
\usepackage{etex}
\usepackage{booktabs}
\usepackage{amsmath,amssymb,amsfonts}
\usepackage{mathtools}
\usepackage{amsthm}
\usepackage{xcolor}
\usepackage{mathtools}
\usepackage{macros}
\usepackage{multirow}
\usepackage{multicol}
\usepackage{kantlipsum,cuted}
\usepackage{relsize}
\usepackage{graphicx}
\usepackage{mwe}
  
\everymath{ \displaystyle}
\RequirePackage{etex} 
\usepackage[ruled,vlined]{algorithm2e}
\usepackage{subcaption}
\usepackage{pgfplots}
\usepackage{tikz}
\pgfplotsset{compat=1.14} 
\newcommand{\bx}{{ x}}
\newcommand{\bu}{{ u}}
\newcommand{\by}{{y}}
\newcommand{\bz}{{ z}}
\newcommand{\bv}{{ v}}
\newcommand{\bw}{{ w}}

\IEEEoverridecommandlockouts

\newtheorem{remark}{Remark}
\usepackage{cite}

\usepackage[english]{babel}
\title{\LARGE\bf Scalable Synthesis of Minimum-Information Linear-Gaussian Control by Distributed Optimization}

\begin{document}

\author{Murat Cubuktepe, Takashi Tanaka, and Ufuk Topcu%
	\thanks{ Murat Cubuktepe, Takashi Tanaka, and Ufuk Topcu are with the Department of Aerospace Engineering
and Engineering Mechanics, University of Texas, Austin, 201 E 24th
St, Austin, TX 78712. email: {\tt\small $\{$mcubuktepe, ttanaka, utopcu$\}$@utexas.edu}}}
\maketitle%

\begin{abstract}
We consider a discrete-time linear-quadratic Gaussian control problem in which we minimize a weighted sum of the directed information from the state of the system to the control input and the control cost.  
The optimal control and sensing policies can be synthesized jointly by solving a semidefinite programming problem. 
However, the existing solutions typically scale cubic with the horizon length.
We leverage the structure in the problem to develop a distributed algorithm that decomposes the synthesis problem into a set of smaller problems, one for each time step.
We prove that the algorithm runs in time linear in the horizon length.
As an application of the algorithm, we consider a path-planning problem in a state space with obstacles under the presence of stochastic disturbances.
The algorithm computes a locally optimal solution that jointly minimizes the perception and control cost while ensuring the safety of the path.
The numerical examples show that the algorithm can scale to thousands of horizon length and compute locally optimal solutions.
\end{abstract}


\section{Introduction}

We revisit the problem of minimum-information control of linear-Gaussian systems~\cite{borkar1997lqg,tatikonda2004stochastic,silva2015characterization,yuksel2013jointly,tanaka2017lqg,tanaka2017semidefinite,kostina2019rate}, where the trade-off between the best achievable control performance and the required sensor data rate is studied.
Such a trade-off is relevant to the utility-privacy trade-off in multi-party control systems~\cite{huang2012differentially,manitara2013privacy,mo2016privacy,wang2014entropy}. 
It also plays a crucial role in the network control systems design~\cite{wang2017differential,cortes2016differential,wang2014entropy}.

The work most related to this paper is~\cite{tanaka2017semidefinite},  where the authors formulated the minimum-information linear-Gaussian control problem as a sensor-controller joint design problem.
They showed that the optimal controller can be obtained as the solution to Riccati equations, while the optimal sensor is obtained as the solution to the so-called Gaussian sequential rate-distortion (SRD) problem~\cite{tanaka2017lqg,stavrou2018optimal}.
They further showed that the Gaussian SRD problem can be formulated as a semidefinite programming problem (SDP), which can be solved by interior-point methods in polynomial time~\cite{ben2001lectures,boyd2004convex}.
However, the computation typically requires $\mathcal{O}(T^3)$ time with horizon length $T$ if the structure of the SDP is not exploited, and an interior-point method may not scale to large horizon lengths. 

Our first contribution in this work is to exploit the structure of the Gaussian SRD problem to derive a distributed algorithm, which facilitates the sensor design for a large horizon length.  
We propose a distributed algorithm based on ADMM~\cite{boyd2011distributed,he2000alternating,wang2001decomposition,hong2017linear,o2015adaptive}, which allows us to design sensors on problems with large horizon length. 
The structure in the problem enables us to develop a distributed algorithm that, given a linear time-varying dynamical system and data rate constraints, decomposes the original SDP problem to smaller SDP problems, one for each time step.
The algorithm then solves the smaller SDP problems in parallel at each iteration.
We prove that our algorithm runs in time linear with the horizon length.


Our second contribution is to apply our algorithm to the minimum-sensing path-planning problem, which incorporates the perception cost into a path-planning problem~\cite{paden2016survey,karaman2011sampling,yang2014literature,elbanhawi2014sampling,richards2002spacecraft,richards2015inter}.
A path-planning problem with a similar spirit has been recently considered in~\cite{stefan2020rationally}.
Conventionally, path planning is performed by a global path search (via grid-based algorithms such as A*~\cite{hart1968formal} and randomized algorithms such as RRT~\cite{lavalle2006planning,karaman2011sampling}), followed by path-smoothing and feedback control design for path following.
In this paper, we provide a path-smoothing algorithm that refines a given initial trajectory by computing a locally optimal solution.
Our algorithm is closely related to the convex-concave procedure~\cite{lipp2016variations}, which iteratively finds a locally optimal solution to a difference of convex (DC) problem and to~\cite{stefan2020rationally}, which optimizes a similar metric to our case in a path-planning problem using an RRT*-based method. 
We formulate the path-smoothing problem with perception cost as a DC problem. 
Our formulation is convex, except the task constraints, which is to avoid colliding with obstacles.
We express the task constraints as reverse convex constraints, which can be used in a DC problem, and a locally optimal solution can be computed.

We show the effectiveness of the algorithms in two numerical examples. 
To demonstrate our first contribution, we consider attitude determination for a spin-stabilized satellite example.
We compute the minimum required information in order to satisfy the required distortion constraints with a substantial horizon length. 
To demonstrate the second contribution, we solve a path-smoothing problem with multiple obstacles in a two-dimensional state space. 
The objective is to find a path that is feasible and minimizes the required joint control and perception cost.
We show that we can find a locally optimal path within a few iterations.

\textit{Organization:}  Section II introduces the Gaussian SRD problem and the semidefinite programming formulation for the linear-Gaussian sensor design problem. 
We propose a distributed algorithm based on ADMM in Section III that solves the linear-Gaussian sensor design problem in linear time with the horizon length. 
We discuss the convergence rate of the algorithm and improvements over the standard ADMM.
Section IV provides the application of our algorithm in a path-planning problem.
Section V provides two examples to show the validity of our proposed algorithm.
Section VI concludes the paper and discusses future directions.

\section{Problem Formulation and Preliminaries}

\emph{Notation.} We denote vector-valued variables at time step $t$ by $x_t \in \mathbb{R}^n$.  $\bx(T)$ is the collection of all $x_0\ldots x_T$. $A \succ 0 \;( \succeq 0 )$ denotes that the matrix $A \in \mathbb{R}^{n \times n}$ is positive definite (positive semidefinite). We denote a Gaussian distribution with mean $\mu$ and covariance $\Sigma$ by $\mathcal{N}(\mu,\Sigma)$. The log-determinant of a positive definite matrix $A$ is given by logdet $A$ and the trace of a matrix $A$ is given by Tr$(A)$. For a positive semidefinite matrix A, we define $\| x \|_{A}=\sqrt{x^{\top}Ax}.$ 

\subsection{Minimum-Information Linear-Gaussian Control}

\label{secpsrdformulation}
The model that we consider in this paper is the following linear time-varying system%
\begin{align}
\label{gmprocess}
\bx_{t+1}=A_t \bx_t+B_t u_t + \bw_t, \;\; t=1, 2, \ldots, T-1,
\end{align}%
where $\bx_0 \in \mathbb{R}^n\sim\mathcal{N}(0, P_0), P_0 \in \mathbb{R}^{n \times n}\succ 0$ and $\bw_t\in \mathbb{R}^n\sim \mathcal{N}(0, W_t), W_t \in \mathbb{R}^{n \times n} \succ 0, t=1, 2, \ldots, T$ are mutually independent Gaussian random variables, $u_t \in \mathbb{R}^m, t=1, 2, \ldots, T-1$  are the control inputs, and $A_t \in \mathbb{R}^{n \times n}$, $B_t \in \mathbb{R}^{n \times m}, t=1, 2, \ldots, T-1$ define the dynamics.

 The minimum-information linear-Gaussian control problem is formulated as
\begin{align}
&\text{minimize} \quad \;\; I(\bx(T) \rightarrow u(T)) \label{psrd1}\\
&\text{subject to } \quad \mathbb{E}\|\bx_t\|^2 + \|u_t\|^2 \leq \gamma_t \label{psrd2},
\end{align}%
where the minimization is over the causal policies (stochastic kernels) of the form $\{p(u_t|x(t),u(t-1))\}_{t=1, 2, \ldots, T}$. Positive constants $\gamma_t, t=1, 2, \ldots , T$ specify user-defined requirements on instantaneous control costs, and the term $I(\bx(T)\rightarrow \bu(T))$ is known as \emph{directed information} \cite{massey1990causality}:
\begin{align}
\label{eqdirectedinfo}
I(\bx(T)\rightarrow \bu(T))\triangleq \sum\nolimits_{t=1}^T I(\bx(t);\bu_t|\bu(t-1)),
\end{align}%
where $I(\bx(t);\bu_t|\bu(t-1))$ is the conditional mutual information.
It can be shown \cite{tanaka2017lqg} that the optimal policy for \eqref{psrd1}-\eqref{psrd2} can be realized by a three-stage structure comprised of
\begin{itemize}
    \item[(i)] a linear sensor mechanism%
    \begin{align}
    \label{srdchannel}
    \by_t=C_t\bx_t+\bv_t, \;\; t=1, 2, \ldots, T,
    \end{align}%
    where $\bv_t\sim\mathcal{N}(0,V_t), V_t\succ 0$ are mutually independent Gaussian random variables;
    \item[(ii)] the least mean square error estimator (Kalman filter)%
    \begin{equation}
    \label{eq:KF}
    z_t= \mathbb{E}(x_t|y(t)), \quad t=1, 2, \ldots, T \text{; and}
    \end{equation}%
    \item[(iii)] the certainty equivalence controller $u_t=K_tz_t$.
\end{itemize}
Observing that the optimal control gain $K_t$ in (iii) can be pre-calculated by solving a backward Riccati equation, the problem~\eqref{psrd1}--\eqref{psrd2} can be reduced to the optimal sensor design problem~\cite{tanaka2017lqg}
\begin{align}
&\text{minimize} \quad \;\; I(\bx(T) \rightarrow z(T)) \label{1_psrd1}\\
&\text{subject to } \quad \mathbb{E}\|\bx_t-\bz_t \|^2_{\Theta_t}\leq \gamma_t-c_t (\coloneqq D_t) \label{1_psrd2},
\end{align}%
where the minimization is over the sensor mechanism \eqref{srdchannel} (decision variables are matrices $C_t$ and $V_t$) and $z_t$ is computed by \eqref{eq:KF}. Constants $\Theta_t\succ 0$ and $c_t$ are obtained from the solution to the aforementioned backward Riccati equation.
The problem~\eqref{1_psrd1}--\eqref{1_psrd2} is known as the Gaussian sequential rate-distortion (SRD) problem \cite{tatikonda2004control} (also known as the nonanticipative rate-distortion problem \cite{stavrou2018optimal}).

\subsection{SDP Formulation of the Gaussian SRD Problem}

It can be further shown (see \cite{tanaka2017semidefinite} for the derivation) that the Gaussian SRD problem \eqref{1_psrd1}-\eqref{1_psrd2} can be written as a SDP
\begin{align}
    \displaystyle & \text{minimize}\quad - \sum\nolimits_{t=1}^T \text{logdet}~\Pi_{t}\label{eq:obj}\\
    &\text{subject to} \quad P_{T|T}=\Pi_{T}, \Pi_{t} \succ 0, \quad t=1, 2, \ldots, T,\label{eq:cons1}\\
    &\hspace{-0.20cm}\begin{bmatrix}
    P_{t|t}-\Pi_{t} & \hspace{-0.1cm} P_{t|t}A^{\top}_t\\
    A_t P_{t|t} & \hspace{-0.1cm} W_t + A_t P_{t|t}A^{\top}_t\label{eq:cons2}
    \end{bmatrix}\succeq 0, \; t=1, 2, \ldots,T-1,\\
    &\text{Tr}(\Theta_{t}P_{t|t}) \leq D_t, \quad t=1, 2, \ldots,T-1,\label{eq:cons3}\\
    &P_{t+1|t+1}\preceq A_t P_{t|t}A^{\top}_t +W_t, \quad t=1, 2, \ldots, T-1, \label{eq:cons}
\end{align}%
where $P_{t|t}, \Pi_t \in \mathbb{R}^{n \times n}$ are the problem variables, and $A_t, W_t, D_t, \Theta_t$ are given problem data. Once the optimal solution $\{P_{t|t}\}_{t=1, 2, \ldots , T}$ to \eqref{eq:obj}--\eqref{eq:cons} is obtained, the optimal sensor mechanism \eqref{srdchannel} can be obtained by choosing the matrices $C_t$ and $V_t$ to satisfy
\begin{align*}
    P^{-1}_{t|t}-(A_{t-1} P_{t-1|t-1}A^{\top}_{t-1} +W_{t-1})^{-1}=C_t^{\top}V^{-1}_tC_t.
\end{align*}%

\subsection{Problem Statement}%
Reference~\cite{tanaka2017lqg} notes that solving the SDP~\eqref{eq:obj}--\eqref{eq:cons} typically requires $\mathcal{O}(T^3)$ time. 
However, that generally holds if there is no sparsity pattern in the SDP problem that can be exploited.
We note that the only coupling constraints between different time-steps in the SDP~\eqref{eq:obj}--\eqref{eq:cons} is the constraint in~\eqref{eq:cons}, and we propose a method to decouple these constraints for different time-steps, which facilitates solving the linear-Gaussian sensor design problem in $\mathcal{O}(T)$ time. Specifically, we solve the following problem.%
\begin{problem}
Given the dynamics $A_t,$ the Gaussian noise matrices $W_t,$ and coefficients $D_t$ for $t=1,\ldots,T$, derive an algorithm that solves the SDP~\eqref{eq:obj}--\eqref{eq:cons} in $\mathcal{O}(T)$ time. 
\end{problem}%
In addition to solving Problem~1 in Section~\ref{sec:admm}, we also discuss how we utilize our methods for the Gaussian SRD problem to a path-planning problem in Section~\ref{sec:path}.

\section{Distributed Sensor Design Using ADMM}\label{sec:admm}%
In this section, we derive our distributed algorithm for the linear-Gaussian sensor design problem.
We introduce alternating direction method of multipliers (ADMM)~\cite{boyd2011distributed}, which is frequently used in distributed optimization. 
We then derive our formulation to decouple the constraints in~\eqref{eq:cons} for different time steps.
We prove that by construction, our algorithm also runs in time linear with the horizon length $T$. 

\subsection{ADMM}

ADMM can be used to solve the following constrained optimization problem%
\begin{align*}
    &\text{minimize} \quad\; f(j)\\
    &\text{subject to} \quad j \in \mathcal{C},
\end{align*}
where $j \in \mathbb{R}^n$ is the problem variable, $f$ and $\mathcal{C}$ are convex. We rewrite the problem in ADMM form as%
\begin{align*}
    &\text{minimize} \quad\; f(j)+g(k)\\
    &\text{subject to} \quad j=k,
\end{align*}
where $g$ is the indicator function of $\mathcal{C}$.
 
The augmented Lagrangian with the scaled dual variable $l \in \mathbb{R}^n$ for this problem is%
$$
L_{\rho}(j,k,l)= f(j)+g(k)+(\rho/2)\|j-k+l\|^2_2,
$$
which is obtained by combining the terms in the augmented Lagrangian, see~\cite[Section~3.1.1]{boyd2011distributed} for details. The ADMM iterations for this problem are%
\begin{align*}
&j^{m+1}\coloneqq{\text{argmin}}\lbrace f(j)+(\rho/2)\|j-k^m-l^m\|^2_2\rbrace,\\
&k^{m+1}\coloneqq\pi_{\mathcal{C}}(j^{m+1}+k^m),\\
&l^{m+1}\coloneqq l^m+j^{m+1}-k^{m+1},
\end{align*}
where $\rho \in \mathbb{R_+}$ is a penalty parameter, $j^{m}$ is the value of $j$ after $m'$th iteration and $\pi_{\mathcal{C}}$ denotes projection onto $\mathcal{C}.$

The $j$-update involves minimizing $f$ plus a convex quadratic function, i.e., evaluation of the proximal operator associated with $f$. 
The $k$-update is Euclidean projection onto $C$. 

\subsection{Distributed Linear-Gaussian Sensor Design}

We now give our algorithm that solves the linear-Gaussian sensor design problem in $\mathcal{O}(T)$ time. 
Our algorithm involves constructing an equivalent problem to the problem in~\eqref{eq:obj}--\eqref{eq:cons} and deriving ADMM updates that can be computed in $\mathcal{O}(T)$ time, which proves that we can solve the Gaussian SRD problem in time linear with the horizon length $T$.

\begin{theorem}
The Gaussian SRD problem can be solved in $\mathcal{O}(T)$ time with the horizon length $T$.
\end{theorem}

\begin{proof}
Our proof is based on constructing an equivalent optimization problem that can be solved in a distributed manner with ADMM such that each ADMM update runs in time linear with the horizon length $T$.

We rewrite the SDP~\eqref{eq:obj}--\eqref{eq:cons} by adding additional variables $K_t$, $Q_{t|t},$ and $S_t$, and the additional constraints in~\eqref{eq:admm_couple}--\eqref{eq:admm_cons6},%
\begin{align}
    \displaystyle & \text{minimize}\quad - \sum\nolimits_{t=1}^T \text{logdet}~\Pi_{t}\label{eq:objective}\\
    &\text{subject to} \;\;\;P_{T|T}=\Pi_{T}, \Pi_{t} \succ 0, \;\;\; t=1, 2, \ldots, T,\label{eq:admm_cons1}\\
    &\hspace{-0.2cm}\begin{bmatrix}
    P_{t|t}-\Pi_{t} &\hspace{-0.1cm}P_{t|t}A^{\top}_t\\
    A_t P_{t|t} & \hspace{-0.1cm} W_t + A_t P_{t|t}A^{\top}_t
    \end{bmatrix}\succeq 0,\; t=1, 2, \ldots, T-1,\\
    &\text{Tr}(\Theta_t P_{t|t}) \leq D_t,\;\;\; t=1, 2, \ldots, T-1,\label{eq:admm_cons3}\\
    &P_{t+1|t+1}+\hspace{-0.03cm}S_t\hspace{-0.03cm}=\hspace{-0.03cm}A_t Q_{t|t}A^{\top}_t \hspace{-0.03cm}+\hspace{-0.03cm}W_t,\; t=1, 2, \ldots, T-1,\label{eq:admm_couple}\\
    &K_t\succeq 0,\;\;Q_{t|t}=P_{t|t},\;S_t=K_t,\; t=1, 2, \ldots, T-1\label{eq:admm_cons6}.
\end{align}%

By construction, it is clear that the optimization problem in~\eqref{eq:obj}--\eqref{eq:cons} shares the same objective value and set of optimal solutions with the problem~\eqref{eq:objective}--\eqref{eq:admm_cons6}.

We rewrite the above problem in ADMM form with $j$ denoting the variables for $S_t$ and $Q_{t|t}$, $k$ denoting the variables for $P_{t|t}$, $\Pi_t$, $K_t$, and $f(j)$ is the indicator function for~\eqref{eq:admm_couple} and $g(k)$ is the sum of the objective in~\eqref{eq:objective} and the indicator functions for~\eqref{eq:admm_cons1}--\eqref{eq:admm_cons3} and~$K_t \succeq 0$. 

We now construct the ADMM iterations as follows.
The $j$-update is given by solving the following convex problem%
\begin{align*}
  & \text{minimize}\quad  \sum\nolimits_{t=1}^{T-1} (\|Q_{t|t}-P^m_t+U_t^m\|^2_{F}+ \\\nonumber &\|S_t-K^m_t+V_t^m\|^2_{F})\\
    &\text{subject to} \\
    & P^m_{t+1|t+1}+S_t= A_t Q_{t|t}A^{\top}_t +W_t, \; t=1, 2, \ldots, T-1,
\end{align*}%
with variables $Q_{t|t}$ and $S_t$. $U_t$ denotes the dual variable for the constraint $Q_{t|t}=P_{t|t}$ and $V_t$ denotes the dual variable for the constraint $S_t=K_t$, and $\|.\|_{F}$ denotes the Frobenius norm of a matrix.  
This problem is separable with each time step $t=1, 2, \ldots, T$, meaning that the optimal solution can be obtained by solving for each time step $t$ separately. Therefore, $j-$ update can be done in time linear in $T$. 

Let $\lbrace Q^{m+1}_{t|t}, K^{m+1}_t\rbrace_{t=1, 2, \ldots, T-1}$ be the optimal solution of this convex optimization problem. 
Then, the $k$-update is given by solving the following SDP%
\begin{align*}
  & \text{minimize}\quad   - (2/\rho)\sum\nolimits_{t=1}^T \text{logdet}~\Pi_{t}+\\
  &\sum\nolimits_{t=1}^{T-1} (\|Q^{m+1}_{t|t}-P_t+U_t^m\|^2_{F}+\|S^{m+1}_t-K_t+V_t^m\|^2_{F})\nonumber\\
    &\text{subject to}\quad P_{T|T}=\Pi_{T}, \Pi_{t} \succ 0, \;\;\; t=1, 2, \ldots, T-1,\\
    &\begin{bmatrix}
    P_{t|t}-\Pi_{t} & P_{t|t}A^{\top}_t\\
    A_t P_{t|t} & W_t + A_t P_{t|t}A^{\top}_t
    \end{bmatrix}\succeq 0, \;\;\; t=1, 2, \ldots, T-1,\\
    &\text{Tr}(\Theta_t P_{t|t}) \leq D_t, \;K_t\succeq 0, \;\; t=1, 2, \ldots, T-1,
\end{align*}%
with variables  $\Pi_t$, $P_{t|t}$ and $K_t$ for $t=1, 2, \ldots, T$. 
Note that this problem also is separable with each time step $t$.

Let $ \lbrace \Pi^{m+1}_t, P^{m+1}_t, K^{m+1}_t\rbrace_{t=1, 2, \ldots, T} $ be the optimal solution of this SDP.
Then, the $l$-update is given by%
\begin{align*}
  &  U^{m+1}_t=U^k_t+Q^{m+1}_{t|t}-P^{m+1}_t,\\
  & V^{m+1}_t=V^m_t+K^{m+1}_t-S^{m+1}_t,
\end{align*}%
for $t=1, 2, \ldots, T-1$. The dual update also scales linearly with $T$. This completes our proof.
\end{proof}%

Our algorithm solves the linear-Gaussian sensor design problem, which is the most computationally challenging part of solving the minimum-information linear-quadratic Gaussian problem in time linear with horizon length.
In the following, we discuss the convergence rate of the algorithm and how we can improve the convergence rate of ADMM.

\subsection{Improvements over the Standard ADMM}

ADMM can generate solutions with moderate accuracy after the first few tens of iterations and solve large-scale problems effectively~\cite{boyd2011distributed}. 
However, ADMM can be very slow to converge to a highly accurate solution.
In this section, we discuss some theoretical properties of the proposed algorithm and possible improvements to improve the convergence rate.

\subsubsection{Convergence Rate of the Algorithm}
ADMM can achieve linear convergence rate, i.e., requiring $\mathcal{O}(\log (1/\epsilon))$ iterations to achieve $\epsilon$ accuracy by choosing a small enough step-size in the dual update~\cite{hong2017linear}, or if $f$ is strongly convex~\cite{giselsson2016linear}.


\subsubsection{Using Acceleration Steps}
An extension to ADMM is to use an additional acceleration step, which is including additional update steps in the $k-$update and the dual update. 
It is empirically shown in~\cite{goldstein2014fast} that acceleration steps can significantly improve the convergence rate.  
Adding an acceleration step requires the objective $f$ to be strongly convex to ensure convergence. 
Therefore, additional strongly convex terms are required to be added to the objective in the $j-$ update, such as $\mu\left(\|Q_{t|t}\|^2_{F}+\|S_t\|^2_{F}\right)$, where $\mu$ is a positive parameter.
   
The acceleration step is carried out by modifying the ADMM iterations as%
\begin{align*}
&j^{m+1}\coloneqq{\text{argmin}}\lbrace f(j)+(\rho/2)\|j-\hat{k}^m+\hat{l}^m\|^2_2\rbrace,\\
&k^{m+1}\coloneqq\pi_{\mathcal{C}}(j^{m+1}+\hat{l}^m),\\
&l^{m}\coloneqq\hat{l}^{m}+\rho(j^{m+1}-k^{m+1}),\\
&\beta^{m+1}\coloneqq (1+\sqrt{1+4{\beta^m}^2})/2,\\
&\hat{k}^{m+1}\coloneqq k^m+\dfrac{\beta^m -1}{\beta^{m+1}}(k^m-k^{m-1}),\\
&\hat{l}^{m+1}\coloneqq l^m+\dfrac{\beta^m -1}{\beta^{m+1}}(l^m-l^{m-1}),
\end{align*}%
where $\beta^m$ is the acceleration parameter, and $\beta^0=0.$
Restarting the algorithm here refers to setting the acceleration parameter to $0$ during the iterations.

\subsubsection{Over-relaxation}

Over-relaxation is done by replacing $j^{m+1}$ with $\gamma^m j^{m+1}-(1-\gamma^m)j^m$ in $k-$ and $l-$ updates. $\gamma^m \in (0,2)$ is a relaxation parameter. When $\gamma^m > 1$, this technique is called over-relaxation, and when $\gamma^m < 1$, it is called under-relaxation. 
The intuition in over-relaxation is to take an additional step in the $k-$ and $l-$ updates. References~\cite{eckstein1992douglas,nishihara2015general} analyze the convergence of over-relaxation for different parameters. 
Experiments in~\cite{eckstein1998operator,eckstein1994parallel} suggest that over-relaxation with $\gamma^m \in [1.5,1.8]$ can improve convergence.

\section{Applications to Path-Planning Problem}\label{sec:path}

In this section, we leverage our SDP-based formulation to synthesize an optimal control and sensing policy to the path-smoothing problem.
The path-smoothing problem we consider has a similar structure to~\eqref{eq:obj}--\eqref{eq:cons} with additional variables and constraints.
The ADMM-based algorithms provided in the previous section are largely applicable.
We provide an algorithm for the path-smoothing problem that locally optimizes a weighted sum of the control and perception cost from any given trajectory by iteratively convexifying the problem around the trajectory.
Our algorithm computes a locally optimal solution, and if the given trajectory is feasible, then all successive trajectories are guaranteed to be feasible.

Suppose~\eqref{gmprocess} represents the dynamics of a mobile robot under stochastic disturbances, where $x_t \in \mathbb{R}^n$ for $t=1, 2, \ldots, T$ is the position of the robot, and $u_t \in \mathbb{R}^m$ for $t=1, 2, \ldots, T-1$ is the control input.
Let $X_{\mathrm{obs}} \in \mathbb{R}^m$ be a closed set of points that represents the obstacles to avoid. 
We can formulate the path-planning problem by modifying the objective in~\eqref{eq:obj} by%
\begin{align}
      \displaystyle & \text{minimize}\quad \sum\nolimits_{t=1}^T (-\text{logdet}~\Pi_{t}+(1/\alpha)\| u_t\|^2_2)\label{eq:obj_path}
\end{align}%
and adding the following constraints to~\eqref{eq:cons1}--\eqref{eq:cons},%
\begin{align}
 &u_t \in \mathcal{U}_t, \quad t=1, 2, \ldots, T-1\label{eq:control_cons_path},\\
  &x_t \in \mathcal{X}_t, \quad t=1, 2, \ldots,T,\label{eq:state_cons_path}\\
 &x_{t+1}=A_t x_t+B_t u_t, \quad t=1, 2, \ldots, T-1,\label{eq:dynamics_path}\\
 &(x_t-x_{\mathrm{obs}})^{\top}P^{-1}_{t|t}(x_t-x_{\mathrm{obs}})\geq \chi^2,\label{eq:cons_avoid_path}\\ \nonumber\quad & t=1, 2, \ldots, T, \quad \forall x_{\mathrm{obs}} \in X_{\mathrm{obs}},
\end{align}%
where $\alpha \in \mathbb{R}_{+}$ is a parameter that gives the trade-off between the control and the perception cost, $\chi^2 \in \mathbb{R}_{+}$ is a confidence level parameter,
$A_t \in \mathbb{R}^{n \times n}$ and $B_t \in \mathbb{R}^{n \times m}$ are given matrices, and ${x}_{\mathrm{obs}} \in \mathbb{R}^n$ are the obstacles. 
Here, we are adopting the notion of path-planning in the so-called \emph{uncertain configuration space}~\cite{lambert2003safe}, i.e., designing the sequence of mean $x_t$ and covariance $P_{t|t}$ jointly.

The objective in~\eqref{eq:obj_path} minimizes a trade-off between the perception and control cost. 
The constraints in~\eqref{eq:control_cons_path}--\eqref{eq:state_cons_path} represent the constraints in the inputs and final state. 
The constraints in~\eqref{eq:dynamics_path} give the dynamics of the linear time-varying process, and \eqref{eq:cons_avoid_path} ensures that the uncertainty ellipsoid given by $P_{t|t}$ around the trajectory $x_t$ does not intersect with the obstacles.
The problem in~\eqref{eq:cons1}--\eqref{eq:cons} and \eqref{eq:obj_path}--\eqref{eq:cons_avoid_path} is a convex optimization problem without the constraints in~\eqref{eq:cons_avoid_path}, which is a nonconvex constraint in general. 
  
\begin{remark}
We can include any convex cost, such as in the position or control variables in the objective~\eqref{eq:obj_path}. In this paper, we consider control and perception cost for simplicity.
\end{remark}

We use the penalty convex-concave procedure (CCP)~\cite{lipp2016variations}, which iteratively over-approximates a non-convex optimization problem via linearization. 
The resulting convex problem can then be solved efficiently, and we iterate the process until we compute a locally optimal solution.
Specifically, we compute affine upper bounds for the convex functions in~\eqref{eq:control_cons_path}.
CCP improves the solution by convexifying the problem around the previous solution iteratively.

To perform CCP, we start with any initial trajectory $\lbrace \hat{x}_{t}, \hat{P}_{t|t}\rbrace_{t=1, 2, \ldots, T}$ in the uncertain configuration space and convexify the constraints in~\eqref{eq:cons_avoid_path} for $t=1, 2, \ldots, T, \quad \forall x_{\mathrm{obs}} \in X_{\mathrm{obs}}$ as%
\begin{align}
&x_{t,\mathrm{obs}}=(x_t-x_{\mathrm{obs}}),\nonumber\\
 &   h(x_t,P_{t|t},x_{\mathrm{obs}})=x_{t,\mathrm{obs}}^{\top}P^{-1}_{t|t}x_{t,\mathrm{obs}},\nonumber\\
  &  \bar{h}(x_t,P_{t|t},x_{\mathrm{obs}})=h(\hat{x}_t,\hat{P}_{t|t},x_{\mathrm{obs}})+\label{eq:convexification3}\\
  &\nabla h(\hat{x}_t,\hat{P}_{t|t},x_{\mathrm{obs}})(\hat{x}_t-x_t,\hat{P}_{t|t}-P_{t|t},x_{\mathrm{obs}}),\nonumber\\
  &\nabla h(\hat{x}_t,\hat{P}_{t|t},x_{\mathrm{obs}})=-\hat{P}_{t|t}^{-1}x_{t,\mathrm{obs}}x_{t,\mathrm{obs}}^{\top}\hat{P}_{t|t}^{-1},\label{eq:convexification4}\\
  &\bar{h}(x_t,P_{t|t},x_{\mathrm{obs}})\geq \chi^2,\label{eq:cons_avoid_path_ccp}
\end{align}%
where~\eqref{eq:convexification3} computes an affine approximation of the convex function $h$ around $(\hat{x}_{t}, \hat{P}_{t|t})$ and~\eqref{eq:convexification4} is the gradient of $h$ at $(\hat{x}_{t}, \hat{P}_{t|t})$. \eqref{eq:cons_avoid_path_ccp} is the resulting convexified constraint.  

The function $\bar{h}$ is a \emph{convex over-approximation} of the original function $h$ as we compute an affine lower bound of $h$.
As a direct consequence, any feasible assignment for the resulting over-approximated and convex problem is also feasible for the original nonconvex problem.
However, the resulting convex problem might be infeasible, even though the original problem is not. 
To find a feasible assignment, we assign a non-negative penalty variable $m_{t,\mathrm{obs}}$ for each of the convexified constraints by modifying the constraint~\eqref{eq:cons_avoid_path_ccp} as~
\begin{align}
  m_{t,\mathrm{obs}}+\bar{h}(x_t,P_{t|t},x_{\mathrm{obs}})\geq \chi^2\label{eq:cons_avoid_path_ccp2}
\end{align}
to ensure that the convexified problem is always feasible~\cite[Section 3.1]{lipp2016variations}.
To find a solution that minimizes violations to the convexified constraints, we minimize the sum of the penalty variables. 
We then solve the convex problem with the objective%
\begin{align}
      \displaystyle & \hspace{-0.15cm}\text{minimize}\quad \sum\nolimits_{t=1}^T (-\text{logdet}~\Pi_{t}+(1/\alpha)\| u_t\|^2_2)+\nonumber\\
      &\hspace{-0.15cm}\displaystyle \tau\sum\nolimits_{t=1}^T\sum\nolimits_{x_\mathrm{obs} \in X_\mathrm{obs}}m_{t,\mathrm{obs}} \nonumber
      \end{align}%
      and the constraints in~\eqref{eq:cons1}--\eqref{eq:cons},  \eqref{eq:control_cons_path}--\eqref{eq:dynamics_path} and~\eqref{eq:cons_avoid_path_ccp2}
      where $\tau$ is a positive penalty parameter that minimizes a trade-off between the objective and the violations of the constraints.
If all penalty variables are assigned to zero, then the solution of the convex problem is feasible for the original non-convex path-planning problem, as we over-approximate the convex functions by affine functions. 
If any of the penalty variables $m_{t,\mathrm{obs}}$ and are assigned to a positive value, we update the penalty parameter $\tau$ by $\mu \cdot \tau$ for a $\mu > 0$, similar to~\cite[Algorithm 3.1]{lipp2016variations}.

After getting a new solution, we convexify the non-convex constraints by linearizing the convex functions around the new solution and solve the resulting convex SDP. 
We repeat the procedure until we find a feasible and locally optimal solution. 
If the procedure converges to an infeasible solution, we restart the procedure with another initial trajectory.
Note that the procedure converges to a locally optimal solution for a fixed $\tau$, i.e., after
$\tau=\tau_\mathrm{max}$, but it may converge to an infeasible point of the original problem~\cite[Section 1.3]{lipp2016variations}.

\begin{remark}
If the initial trajectory is feasible, we can set the penalty variables $m_\mathrm{t,obs}$ to zero, and all subsequent trajectories are guaranteed to be feasible after each iteration of the CCP~\cite{lipp2016variations}.
\end{remark}




\section{Numerical Examples}

\begin{figure}[t]
\input{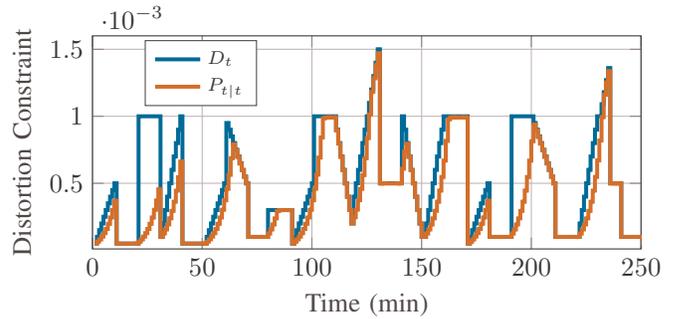}
\caption{Plot of the time-varying distortion constraint $D_{t}$ and the traces of the optimal covariance matrices  $P_{t|t}$ for the first 250 minutes of the mission.}
\label{fig:dist}
\end{figure}

 \begin{figure}[t]
 \input{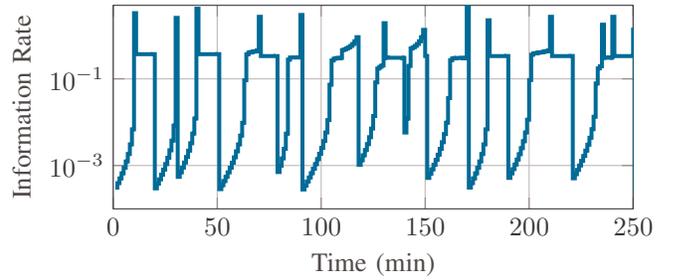}
 \caption{Plot of the optimal information rate $I(\bx(t):\by_t|\by(t-1))$ for the first 250 minutes of the mission.}
 \label{fig:info}
 \end{figure}

We evaluate our distributed sensor and control design procedure in two numerical examples.
The experiments are performed on a computer with an Intel Core i9-9900u 2.50 GHz processor and 64 GB of RAM with Mosek~\cite{aps2017mosek} as the SDP solver. 
The first example is on computing the minimal information to estimate the attitude of a satellite subject to accuracy constraints.
The second example is on a path-smoothing problem in a two-dimensional state space containing multiple obstacles in order to illustrate the effects of varying the trade-offs between the control and perception cost.


\subsection{Satellite Attitude Determination}

In this example, we consider the spin-stabilized satellite example from~\cite{tanaka2017semidefinite}. The equation of motion of the angular velocity vector of a spin-stabilized satellite linearized around the nominal angular velocity vector $(\omega_0, 0, 0)$ is
\[
\left[ \hspace{-0.1cm} \begin{array}{c} d\boldsymbol\omega_1 \\ d\boldsymbol\omega_2 \\ d\boldsymbol\omega_3 \end{array} \hspace{-0.1cm}\right]\hspace{-0.1cm}=\hspace{-0.1cm}
\left[ \hspace{-0.1cm}\begin{array}{ccc} 1 &0&0 \\\hspace{-0.1cm}
0 & 1 & \frac{I_3-I_1}{I_2}\omega_0 \\\hspace{-0.1cm}
0 & \frac{I_1-I_2}{I_3}\omega_0 & 1 \hspace{-0.1cm}\end{array}\right]\hspace{-0.1cm}
\left[\hspace{-0.1cm} \begin{array}{c} \boldsymbol\omega_1 \\ \boldsymbol\omega_2 \\ \boldsymbol\omega_3 \end{array}\hspace{-0.1cm}\right]dt+d{\bf b}
\]
where ${\bf b}$ is a disturbance, $I_1, I_2$, and $I_3$ are the moment of inertias in three dimensions. We convert the continuous-time dynamics into a discrete-time model in the experiments. 
The objective is to figure out the sensing pattern that uses minimal information during the mission. 
We consider a horizon length $T=1500$ with the distortion rate constraint given in Figure~\ref{fig:dist}.

\begin{figure}[t]
\input{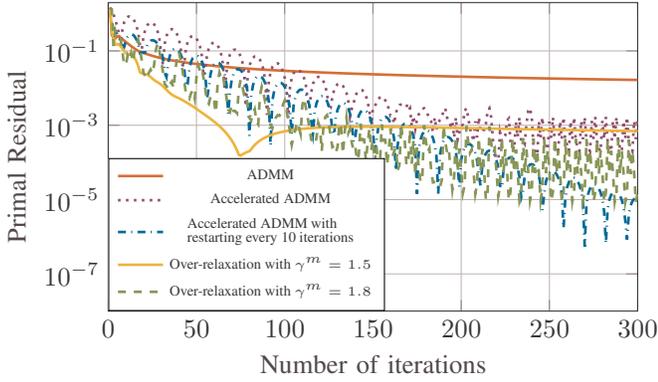}
\caption{Primal residuals for the spin-stabilized satellite example for standard ADMM, accelerated ADMM, and over-relaxed ADMM with $\rho=1$. Accelerated ADMM with restart achieves the best performance.}
\label{fig:prires}
\end{figure}

\begin{figure}[t]
\input{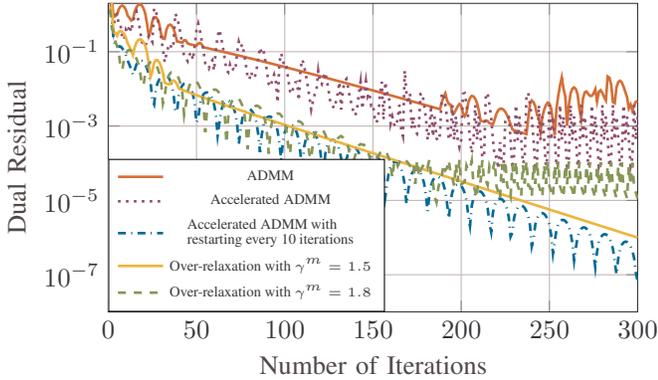}
\caption{Dual residuals for the spin-stabilized satellite example for standard ADMM, accelerated ADMM, and over-relaxed ADMM with $\rho=1$. Accelerated ADMM with restart achieves the best performance.}
\label{fig:dualres}
\end{figure}

We plot the optimal covariance scheduling $P_{t|t}$ and the optimal information rate in Figures~\ref{fig:dist} and~\ref{fig:info} for the first $250$ time-steps.
We put a log-scale on the $y-$axis in Figure~\ref{fig:info} to emphasize that the optimal covariance scheduling requires some information rate in all time steps, which is not the case in the numerical examples in~\cite{tanaka2017semidefinite}. 
The reason that there is a requirement of some information in all time steps may stem from having a less accurate solution from an ADMM-based algorithm rather than from an interior-point method.
We can see that the information rate varies significantly during the mission, and the optimal information rate is minimal if the distortion rate constraint is not restrictive.%
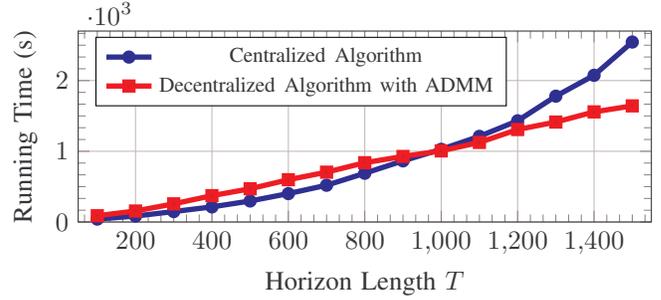
\begin{figure}[t]
\definecolor{color1}{rgb}{0.00000,0.44700,0.74100}%
\definecolor{color2}{rgb}{0.85000,0.32500,0.09800}%
\begin{tikzpicture}
\begin{axis}[%
width=3.00in,
height=1.0in, 
at={(2.167in,0.898in)},
scale only axis,
xmin=50,
xmax=1550,
    minor tick num=5,
xmajorgrids=true,
ymajorgrids=true,
scaled y ticks=base 10:-3,
    major grid style={line width=.2pt,draw=gray!50},
legend style={at={(0.03,0.97)},anchor=north west,font=\fontsize{8}{8}\selectfont},
xlabel style={font=\color{white!0!black}},
xlabel={Horizon Length $T$},
ymin=0,
ymax=2700,
ylabel style={font=\color{white!0!black}},
ylabel={Running Time (s)},
axis background/.style={fill=white},
title style={font=\bfseries},
    point meta=explicit symbolic,
title={},
mark size=1.5pt,
mark=circle,
]
\addplot+ [line width=2.0pt] table[x=x,y=y] {datacent.dat}; 

\addlegendentry{Centralized Algorithm}
\addplot+ [line width=2.0pt] table[x=x,y=y]{dataadmm.dat}; 
\addlegendentry{Decentralized Algorithm with ADMM}

\end{axis}
\end{tikzpicture}
\caption{The running times of the distributed and centralized algorithms with a different number of horizon lengths after 50 iterations with the distributed algorithm.
The distributed algorithm scales linearly with the number of horizon length. 
The centralized algorithm is slower than the distributed algorithm with an increasing number of horizon lengths and does not scale linearly.}
\label{fig:admmplot}
\end{figure}%

We now show the convergence rate of the ADMM-based algorithms by plotting the norm of the primal residual, which is given by $r^1_t=P_t-Q_t$ and $r^2_t=S_t-K_t$ and the dual residual, which is $d^1_t=\rho(P_{t}-P_{t-1})$ and $d^2_t=\rho(K_{t}-K_{t-1})$ for different methods~\cite[Section~3.3]{boyd2011distributed}. 
The primal residual denotes the infeasibility of the methods in each iteration, and the dual residual denotes the optimality of the methods. 
If the primal residual is small, then the variables approach to a feasible solution. 
If the dual residual is small, then the variables approach to an optimal solution.%
\begin{figure}[t]
\input{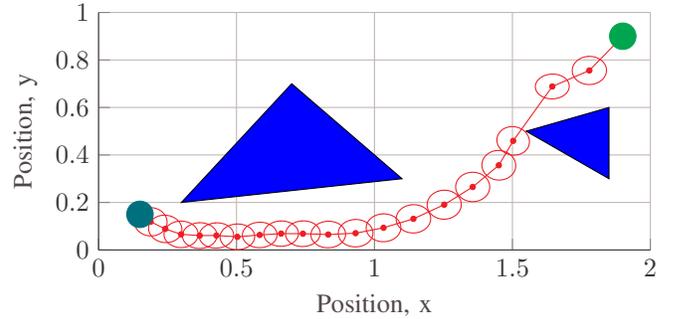}
\caption{The plot of the initial trajectory with covariance matrices. The small red dots denote the expected state at different time steps, and the red ellipsoids are $\chi^2$ covariance ellipses representing $90\%$ certainty regions. The blue regions and the boundaries are the obstacles to avoid.}
\label{fig:init}
\end{figure}%

We list the convergence rates of the over-relaxation and accelerated variants of the ADMM in Figures~\ref{fig:prires}--\ref{fig:dualres}. We plot the primal residuals, in Figure~\ref{fig:prires} we plot the dual residuals in Figure~\ref{fig:dualres}. Residuals for the spin-stabilized satellite example for regular ADMM and over-relaxed ADMM with $\rho=1$.

We note that the over-relaxation and the accelerated variants of the ADMM achieve higher accuracy than the standard ADMM. 
However, the tail convergence of the methods can be very slow. 
Accelerated ADMM with restarting the acceleration parameter in every ten iterations achieves feasibility with higher accuracy compared to other methods, and the tail convergence is faster compared to the other methods. 

For dual residuals, both of the over-relaxation schemes achieve better accuracy. However, the iterates change very slowly before a high accuracy in the primal residual is achieved, which is not a desired property. On the other hand, accelerated ADMM with restarting the acceleration parameter in every ten iterations achieves very high accuracy in the dual residual, and the resulting solution has very high accuracy.

We show the running time of the distributed and the centralized algorithm, which is solving the SDP in~\eqref{eq:obj}--\eqref{eq:cons} directly, for the different number of horizon lengths in Figure~\ref{fig:admmplot}. The results in Figure~\ref{fig:admmplot} show that the running time of the distributed algorithm scales linearly with the horizon length and is faster than the centralized algorithm with a higher horizon length.%
\begin{figure}[t]
\input{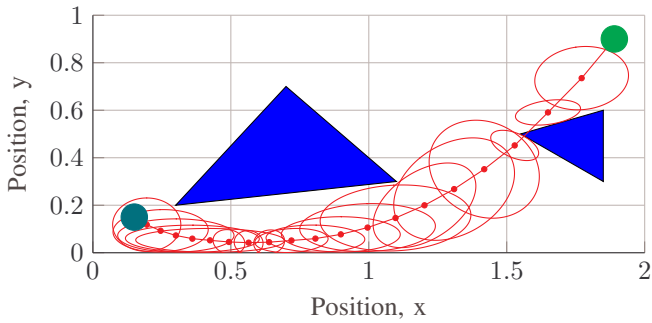}
\caption{Resulting trajectory after 50 CCP iterations with $\alpha=0.1$. The expected values of the locations at each time step are feasible. However, the trajectory may collide with the obstacles as the direct path between the locations in different time-steps crosses one of the obstacles.} 
\label{fig:final001}
\end{figure}%
\begin{figure}[t]%
\input{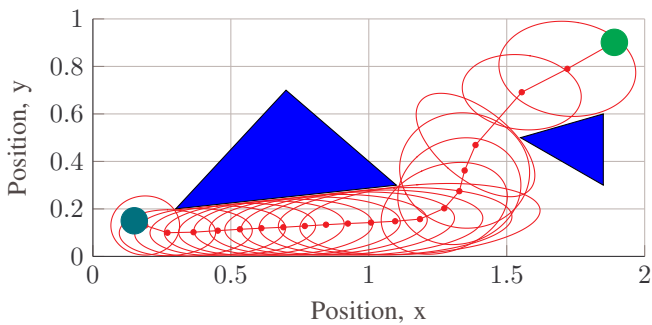}
\caption{Resulting trajectory after 50 CCP iterations with $\alpha=10$. The procedure maximizes the distance between the obstacles and the trajectory by minimizing the perception cost.}
\label{fig:final1}
\end{figure}%
\subsection{Path Planning with Multiple Obstacles}%
We consider a path-planning problem with multiple obstacles in a two-dimensional state space.
We consider an example with $\alpha=10, 1, 0.1, 0,001$ to illustrate the effects of varying the trade-offs between the perception and control cost.
%
\begin{figure}[t]

\begin{tikzpicture}[thick,scale=1.0, every node/.style={scale=1.0}]
\definecolor{color0}{rgb}{0.1215,0.4666,0.7058}
\definecolor{color1}{rgb}{1,0.4980,0.0549}
\definecolor{color2}{rgb}{0.1725,0.6274,0.1725}
\definecolor{color3}{rgb}{0.8392,0.1529,0.1568}
\begin{axis}[
axis y line=right,
axis line style={-},
tick align=outside,
x grid style={white!69.01960784313725!black},
xmajorgrids,
height=2.016in,
width=3.15in,
xmin=-2.45, xmax=51.45,
xtick pos=left,
label style={font=\normalsize},
xtick style={color=black},
y grid style={white!69.01960784313725!black},
ylabel={Control Cost},
ytick style={color=black,font=\normalsize},
ymajorgrids,
ymin=6.43657, ymax=14.18563,
ytick pos=right,
ytick style={color=black}
]
\addplot [ultra thick, color0, dashed]
table {%
0 10.0958
1 7.7105
2 7.2451
3 7.0599
4 7.0909
5 7.0737
6 7.0511
7 7.0245
8 6.9957
9 6.9649
10 6.9334
11 6.9014
12 6.8699
13 6.8386
14 6.8068
15 6.7943
16 6.7899
17 6.7897
18 6.7895
19 6.7895
20 6.7904
21 6.7914
22 6.7928
23 6.7918
24 6.7907
25 6.7911
26 6.7901
27 6.792
28 6.7925
29 6.7931
30 6.7901
31 6.7918
32 6.7931
33 6.789
34 6.7903
35 6.7917
36 6.7907
37 6.7923
38 6.7897
39 6.7899
40 6.7901
41 6.7911
42 6.7902
43 6.7902
44 6.79
45 6.7926
46 6.7888
47 6.7907
48 6.7926
49 6.7906
};
\addplot [ultra thick, color1, dashed]
table {%
0 10.0958
1 7.1825
2 6.913
3 6.9179
4 7.9764
5 8.3243
6 8.4802
7 8.5988
8 8.673
9 8.7152
10 8.7464
11 8.7559
12 8.7788
13 8.7792
14 8.7886
15 8.7879
16 8.783
17 8.7774
18 8.7734
19 8.7651
20 8.7672
21 8.7693
22 8.7694
23 8.7693
24 8.7691
25 8.769
26 8.769
27 8.769
28 8.769
29 8.769
30 8.769
31 8.769
32 8.769
33 8.769
34 8.769
35 8.769
36 8.769
37 8.769
38 8.769
39 8.769
40 8.769
41 8.769
42 8.769
43 8.769
44 8.769
45 8.769
46 8.769
47 8.769
48 8.769
49 8.769
};
\addplot [ultra thick, color3, dashed]
table {%
0 10.0958
1 8.5471
2 8.8114
3 9.2896
4 12.1409
5 13.8274
6 13.8334
7 13.8189
8 13.8185
9 13.8184
10 13.8184
11 13.8184
12 13.8184
13 13.8184
14 13.8184
15 13.8183
16 13.8183
17 13.8183
18 13.8183
19 13.8183
20 13.8183
21 13.8183
22 13.8183
23 13.8183
24 13.8183
25 13.8183
26 13.8183
27 13.8183
28 13.8183
29 13.8183
30 13.8183
31 13.8183
32 13.8183
33 13.8183
34 13.8183
35 13.8183
36 13.8183
37 13.8183
38 13.8183
39 13.8183
40 13.8183
41 13.8183
42 13.8183
43 13.8183
44 13.8183
45 13.8183
46 13.8183
47 13.8183
48 13.8183
49 13.8183
};
\addplot [ultra thick, color2, dashed]
table {%
0 10.0958
1 9.8094
2 11.5749
3 11.8282
4 11.3708
5 11.3807
6 11.5204
7 11.5148
8 11.5235
9 11.5317
10 11.5385
11 11.5442
12 11.5488
13 11.5526
14 11.5557
15 11.5582
16 11.5602
17 11.5619
18 11.5632
19 11.5643
20 11.5652
21 11.5659
22 11.5665
23 11.567
24 11.5674
25 11.5678
26 11.5681
27 11.5684
28 11.5686
29 11.5688
30 11.5689
31 11.5691
32 11.5691
33 11.5691
34 11.5691
35 11.5691
36 11.5691
37 11.5691
38 11.5691
39 11.5691
40 11.5691
41 11.5691
42 11.5691
43 11.5691
44 11.5691
45 11.5691
46 11.5691
47 11.5691
48 11.5691
49 11.5691
};
\end{axis}

\begin{axis}[
legend cell align={left},
height=2.016in,
width=3.15in,
legend columns=2,
legend style={at={(0.98,0.929)},fill opacity=1, draw opacity=1, text opacity=1, draw=white!0.0!black,font=\fontsize{8.0}{8.0}\selectfont},
tick align=outside,
tick pos=left,
label style={font=\normalsize},
x grid style={white!69.01960784313725!black},
xlabel={Number of CCP Iterations},
xmin=-2.45, xmax=51.45,
xtick style={color=black,font=\normalsize},
y grid style={white!69.01960784313725!black},
ylabel={Information Cost},
ymin=22.43301, ymax=69.70299,
ytick={20,30,40,50,60,80},
ytick style={color=black,font=\normalsize}
]
\addplot [ultra thick, color0]
table {%
0 67.5089
1 60.2185
2 52.9994
3 48.7614
4 46.3896
5 45.6051
6 45.324
7 45.1226
8 44.957
9 44.8781
10 44.9154
11 44.9513
12 44.9686
13 44.9947
14 45.0278
15 45.0283
16 45.2237
17 45.2307
18 45.2123
19 45.2801
20 45.2232
21 45.255
22 45.2309
23 45.2319
24 45.2287
25 45.2449
26 45.2251
27 45.2304
28 45.2314
29 45.1949
30 45.2311
31 45.2312
32 45.1842
33 45.2259
34 45.2892
35 45.2119
36 45.286
37 45.2008
38 45.2178
39 45.2858
40 45.2323
41 45.2366
42 45.2349
43 45.2392
44 45.3258
45 45.2021
46 45.2519
47 45.2682
48 45.1994
49 45.2294
};
\addlegendentry{$\alpha=0.01$}
\addplot [ultra thick, color1]
table {%
0 67.5089
1 54.5153
2 43.0517
3 35.9549
4 31.8805
5 30.2494
6 29.6284
7 29.428
8 29.3514
9 29.31
10 29.2832
11 29.2691
12 29.2545
13 29.2465
14 29.2482
15 29.2496
16 29.2514
17 29.2568
18 29.2619
19 29.2964
20 29.2815
21 29.2687
22 29.2663
23 29.2649
24 29.2649
25 29.2651
26 29.265
27 29.2649
28 29.2649
29 29.2648
30 29.2648
31 29.2648
32 29.2648
33 29.2647
34 29.2647
35 29.2647
36 29.2647
37 29.2647
38 29.2647
39 29.2647
40 29.2647
41 29.2647
42 29.2647
43 29.2647
44 29.2647
45 29.2647
46 29.2647
47 29.2647
48 29.2647
49 29.2647
};
\addlegendentry{$\alpha=0.1$}
\addplot [ultra thick, color2]
table {%
0 67.5089
1 51.9793
2 39.9627
3 32.332
4 28.4159
5 26.8942
6 26.6067
7 26.5741
8 26.5627
9 26.5553
10 26.5474
11 26.5424
12 26.5361
13 26.5271
14 26.5133
15 26.492
16 26.4722
17 26.4621
18 26.4604
19 26.46
20 26.4598
21 26.4598
22 26.4599
23 26.4599
24 26.46
25 26.46
26 26.4601
27 26.4601
28 26.4601
29 26.4601
30 26.4601
31 26.4602
32 26.4602
33 26.4602
34 26.4602
35 26.4602
36 26.4602
37 26.4602
38 26.4602
39 26.4602
40 26.4602
41 26.4602
42 26.4602
43 26.4602
44 26.4602
45 26.4602
46 26.4602
47 26.4602
48 26.4602
49 26.4602
};
\addlegendentry{$\alpha=1$}
\addplot [ultra thick, color3]
table {%
0 67.5089
1 51.2356
2 39.1105
3 31.4578
4 27.7488
5 26.354
6 25.7649
7 25.4281
8 25.1733
9 24.923
10 24.7586
11 24.6528
12 24.5961
13 24.5511
14 24.4756
15 24.3375
16 24.218
17 24.0994
18 23.9516
19 23.9298
20 23.928
21 23.9168
22 23.8182
23 23.7188
24 23.6438
25 23.6271
26 23.6273
27 23.6283
28 23.6298
29 23.6312
30 23.6325
31 23.6336
32 23.6347
33 23.6356
34 23.6364
35 23.6371
36 23.6377
37 23.6383
38 23.6388
39 23.6393
40 23.6396
41 23.64
42 23.6403
43 23.6405
44 23.6408
45 23.641
46 23.6412
47 23.6413
48 23.6414
49 23.6416
};
\addlegendentry{$\alpha=10$}
\end{axis}

\end{tikzpicture}
\vspace{-0.3cm}
 \caption{Perception and control costs versus the number of iterations with different values of $\alpha$. The iterates converge to a locally optimal solution with different parameters, and the resulting perception cost is lower with a higher value of $\alpha.$}
 \label{fig:cost}
\end{figure}
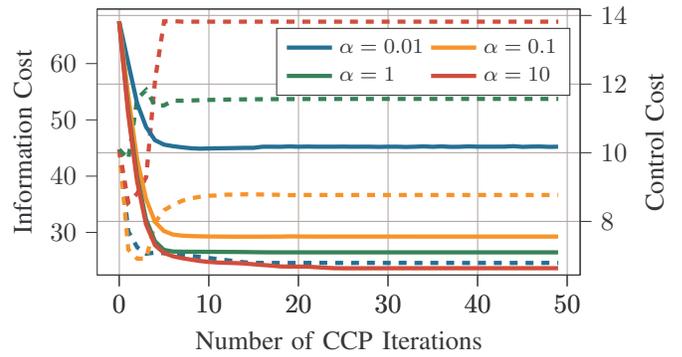%
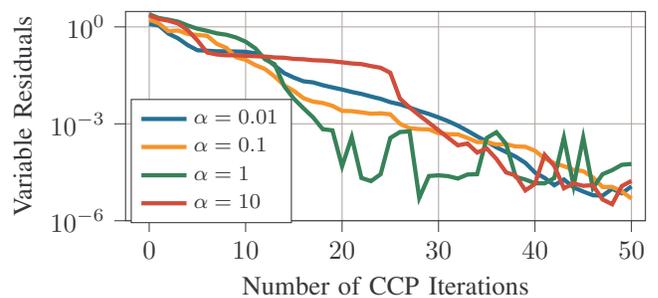
\begin{figure}[t]

\begin{tikzpicture}
\definecolor{color0}{rgb}{0.1215,0.4666,0.7058}
\definecolor{color1}{rgb}{1,0.4980,0.0549}
\definecolor{color2}{rgb}{0.1725,0.6274,0.1725}
\definecolor{color3}{rgb}{0.8392,0.1529,0.1568}
\begin{axis}[%
width=2.721in,
height=1.096in,
yminorticks=true,
ymajorticks=true,
xmajorgrids=true,
ymajorgrids=true,
yminorticks=true,
at={(2.6in,1.15in)},
scale only axis,
legend cell align={left},
legend style={at={(0.32,0.582)},fill opacity=1, draw opacity=1, text opacity=1, draw=white!0.0!black,scale=0.40,font=\fontsize{8.0}{8.0}\selectfont},
tick align=outside,
tick pos=left,
x grid style={white!69.01960784313725!black},
xlabel={Number of CCP Iterations},
xmin=-2.45, xmax=51.45,
xtick style={color=black,font=\normalsize},
y grid style={white!69.01960784313725!black},
ylabel={Variable Residuals},
ymin=1e-6, ymax=3,
ymode=log,
ytick style={color=black,font=\normalsize}
]
\addplot [ultra thick, color0]
table {%
0 1.2307216651178
1 1.10711536023209
2 0.619089681357247
3 0.442042563267115
4 0.272180438410394
5 0.189550517651964
6 0.183376312078445
7 0.175917763519581
8 0.175706734545704
9 0.173166789353357
10 0.169407710322336
11 0.152695392402206
12 0.111773775572462
13 0.063829551596295
14 0.0370973003119783
15 0.0269393761779017
16 0.0207028693268269
17 0.0191464327981978
18 0.0156735479110973
19 0.0133906567298721
20 0.0115182667342055
21 0.00968696790722545
22 0.00843770256196405
23 0.00700176784749854
24 0.00592287002687849
25 0.00468502683947604
26 0.00400608407882659
27 0.00318777967869138
28 0.00248732339800586
29 0.00201923739341464
30 0.00157892664809706
31 0.00120546017943892
32 0.00087762708298466
33 0.00060724850310307
34 0.000408857536004599
35 0.000282270131699741
36 0.000187316795628658
37 0.000134333882591322
38 9.62818446922562e-05
39 6.18653058366634e-05
40 3.11074235467803e-05
41 2.1126716649566e-05
42 1.25688817827023e-05
43 1.89432798289324e-05
44 1.07809382615791e-05
45 8.00715681919682e-06
46 6.16683161393427e-06
47 5.9383919360145e-06
48 9.74243449699378e-06
49 7.08073866569718e-06
50 1.1533286870781e-05
};
\addlegendentry{$\alpha=0.01$}
\addplot [ultra thick, color1]
table {%
0 1.67706840271663
1 1.29475772254854
2 0.73731337083166
3 0.774102906218317
4 0.602832983189601
5 0.562703901010092
6 0.531784511832256
7 0.300560748333764
8 0.224812165783118
9 0.115463793573171
10 0.0957435267090528
11 0.0658363360676252
12 0.0498123571861301
13 0.029312925551922
14 0.0155163536379584
15 0.0104224166461657
16 0.00621609792996509
17 0.00486368150934295
18 0.0044036714907785
19 0.00360821112897813
20 0.00251056405592333
21 0.00246842689535469
22 0.00226569711017494
23 0.00205031991927102
24 0.00211320673109007
25 0.00196555821302357
26 0.00100347788285395
27 0.000730283877458868
28 0.000685089776226013
29 0.000674503764648024
30 0.000511604890585563
31 0.00047717774373829
32 0.000471509999853735
33 0.000365276117465836
34 0.00027283279579954
35 0.000273823721387278
36 0.000227787867316276
37 0.000233384683693736
38 0.000199847593098164
39 0.000194560708574407
40 0.000146625964168022
41 8.46060399997367e-05
42 6.06379948802998e-05
43 4.87862139139958e-05
44 3.77144060754167e-05
45 3.42025492743754e-05
46 2.2112826302385e-05
47 1.10628347310465e-05
48 1.11333164042106e-05
49 8.0151342159158e-06
50 4.86859870082768e-06
};
\addlegendentry{$\alpha=0.1$}
\addplot [ultra thick, color2]
table {%
0 2.47102924606693
1 1.88279891740212
2 1.67279824480276
3 1.45060027307265
4 1.10588410752213
5 0.877597685887689
6 0.757920405303912
7 0.631932669948859
8 0.571995397889186
9 0.471300280717288
10 0.348374543321834
11 0.21793438430068
12 0.0996897993408157
13 0.0701173908735866
14 0.0145181886347833
15 0.00596839780939061
16 0.00302883722806321
17 0.00157774146582956
18 0.000681083229205822
19 0.000627114528518482
20 4.81788702362935e-05
21 0.000389475323974825
22 2.11473841743986e-05
23 1.68458982340067e-05
24 2.75550893604339e-05
25 0.000366197785194207
26 0.000557307569581609
27 0.000584207084825693
28 5.07090601475282e-06
29 2.41922998289384e-05
30 2.54624664315367e-05
31 2.39897165414053e-05
32 2.02235592397942e-05
33 2.69492747413916e-05
34 2.75583284118936e-05
35 0.000365404386393925
36 0.000548452916795528
37 0.000239020172850927
38 2.28022985031154e-05
39 1.87917985559346e-05
40 1.44864960693189e-05
41 1.43414519714486e-05
42 2.09374687387948e-05
43 0.000375353279328564
44 1.20327034673739e-05
45 0.000373973863496843
46 1.02033960069318e-05
47 2.74963939119643e-05
48 3.67538784540166e-05
49 5.38656829254334e-05
50 5.71434021202681e-05
};
\addlegendentry{$\alpha=1$}
\addplot [ultra thick, color3]
table {%
0 2.2507271152214
1 1.75776902255246
2 1.55569799237748
3 1.27197308987788
4 0.786228633740892
5 0.382156909790082
6 0.161460456477553
7 0.140167188083682
8 0.132065589585905
9 0.132373764579543
10 0.127306121054055
11 0.123710206159672
12 0.120833741730705
13 0.1134084178235
14 0.111242793259443
15 0.102204895645197
16 0.0996779424334193
17 0.0929597408886099
18 0.0909473726489974
19 0.0864208077715612
20 0.0802557801689899
21 0.0738382344603286
22 0.0706852969628426
23 0.061241709479437
24 0.0542975558733214
25 0.0379753730669815
26 0.00602664029525135
27 0.00321158184193016
28 0.00183874591463115
29 0.00107851373928053
30 0.000631005281628552
31 0.000370232575154267
32 0.000215864203244229
33 0.000245995229919084
34 0.000128964375539228
35 0.000172338816013724
36 8.19973060692735e-05
37 3.13930938577116e-05
38 2.50248649582234e-05
39 8.67507401447001e-06
40 1.3790822821334e-05
41 0.000110064859451316
42 5.00042912002744e-05
43 1.00872175561936e-05
44 1.42124392473847e-05
45 1.20443582569006e-05
46 1.26239222041274e-05
47 4.56549131769417e-06
48 3.23052239913781e-06
49 1.1783265064482e-05
50 1.71188714364055e-05
};
\addlegendentry{$\alpha=10$}
\end{axis}
\end{tikzpicture}
 \caption{Convergence rate of the procedure with different values of $\alpha$. All problems converge to a solution within an accuracy of $10^{-4}$ after 50 CCP iterations.}
 \label{fig:conv}
\end{figure}%

We plot the initial trajectories in Figure~\ref{fig:init}, where the red dots represent the expected position vector $x_t$ at time step $t$ and the red $\chi^2$ covariance ellipsoids represent 90\% certainty regions. 
Larger covariance ellipsoids indicate a lower perception cost.
We depict the initial and final points with teal and green dots.
We consider the boundaries as obstacles. 
The initial trajectory requires a very high perception cost as the covariance matrices around the trajectory is small in magnitude. 
In this example, we consider $W_t=10^{-2}I$ for all time-steps $t$.

The resulting trajectories after 50 CCP iterations are shown in Figures~\ref{fig:final001} and~\ref{fig:final1}. The trajectories are feasible, and also has a much larger uncertainty around the waypoints, which minimizes the perception cost.
However, the shortest path between the points of the trajectory with $\alpha=0.01$ is infeasible as it crosses the obstacles, and the resulting trajectory may be infeasible in practice with a higher probability than the one with $\alpha=1$. We also demonstrate both of the trajectories in the video\footnote{https://bit.ly/2UvCi26} to demonstrate the trade-offs of having different perception costs on a ground robot. The robot can follow the trajectory with $\alpha=1$ safely, whereas the trajectory with $\alpha=0.01$ results in the robot colliding with the obstacle.

To assess the robustness of the trajectories, we run a-posteriori Monte-Carlo analysis with different trade-offs of perception cost.
We perturb each element of the matrix $A_t$ with $N=10^6$ samples drawn from a zero-mean uniform distribution with an interval length of $10^{-2}$. 
If we apply the exact inputs that we obtained with $\alpha=0.01$ on the perturbed dynamics, the estimated probability of resulting path being infeasible is $0.2827$, even though the nominal trajectory is feasible.
On the other hand, if we apply the inputs with $\alpha=1$ on the same perturbed dynamics, the estimated probability of an infeasible path is $0.0234$, which results in a nominal trajectory that is more robust to the modeling errors.

The resulting perception and control costs with different trade-offs are shown in Figure~\ref{fig:cost}. 
The solid lines depict the perception costs, and the dashed lines depict the control costs. 
We observe that with all different trade-offs of the perception and control costs, the methods converge to a locally optimal solution within a few tens of iterations. 
As expected, the resulting perception costs decrease, and the control cost increase with an increasing value of $\alpha.$
We also plot the convergence rate with different values of $\alpha$ in Figure~\ref{fig:conv}, where $x$-axis shows the difference between the obtained state values $x$ and covariance values $P$ to the converged solution at each iteration, and the slopes of the figures indicate the rate of convergence. 
All problems with different values of $\alpha$ converge to a solution within an accuracy of $10^{-4}$ after $50$ iterations.

\section{Conclusions and Future Work}

We developed a distributed algorithm for solving the minimum-information linear-Gaussian control problem and applied the algorithm to the minimum-sensing path-planning problem. 
Our algorithm can scale to very large horizon lengths and runs in time linear with the horizon length. 
Future work involves establishing the optimal convergence rate given the problem data and determining the optimal parameter selection in different ADMM methods to achieve a faster convergence rate. 
We will also consider nonlinear and multi-agent systems, and integrate our approach to sampling-based methods for path-planning problems.

\bibliographystyle{ieeetr}
\bibliography{literature}

\end{document}